\documentclass[12pt]{article}
\usepackage{amsmath,amsthm,amsfonts}
\usepackage{fullpage}

\newtheorem{theorem}{Theorem}

\newtheorem{lemma}[theorem]{Lemma}

\title{Detecting patterns in finite regular and context-free languages}

\author{Narad Rampersad\\
Department of Mathematics and Statistics\\
University of Winnipeg\\
Winnipeg, MB  R3B 2E9, Canada\\
{\tt n.rampersad@uwinnipeg.ca}\medskip\\
Jeffrey Shallit\\
School of Computer Science\\
University of Waterloo\\
Waterloo, ON  N2L 3G1, Canada\\
{\tt shallit@graceland.uwaterloo.ca}}

\begin{document}
\date{\today}
\maketitle

\begin{abstract}
We consider variations on the following problem: given an NFA $M$ over the
alphabet $\Sigma$ and a pattern $p$ over some alphabet $\Delta$,
does there exist an $x \in L(M)$ such that
$p$ matches $x$?  We consider the restricted problem where $M$ only
accepts a finite language.  We also consider the variation where the
pattern $p$ is required only to match a factor of $x$.  We show that
both of these problems are NP-complete.  We also consider the same problems
for context-free grammars; in this case the problems become PSPACE-complete.
\end{abstract}

\section{Introduction}

    The computational complexity of pattern matching has received much
attention in the literature.  Although determining whether a given word
appears inside another can be done in linear time, other pattern-matching
problems appear
to be computationally intractable.  In a classic paper, Angluin \cite{Ang80} 
showed the problem of determining if an arbitrary pattern matches an
arbitrary string is NP-complete.  More recently, 
Anderson et al.\ \cite{ALR+09} showed that pattern matching becomes
PSPACE-complete if we are trying to match a pattern against words of a
language specified by a DFA, NFA, or regular expression.  

     In this paper we consider some variations on the pattern matching
problem.  We begin by fixing our notation.

Let $\Sigma$ be an alphabet, i.e., a nonempty, finite set
of symbols (letters). By $\Sigma^*$ we denote the set of
all finite words (strings of symbols) over $\Sigma$, and by
$\epsilon$, the empty word (the word having zero
symbols).  If $w = xyz$, then $y$ is said to be a \emph{factor} of $w$.

Let $k \geq 2$ be an integer.  A word $y$ is a
\emph{$k$-power} if $y$ can be written as $y = x^k$ for
some non-empty word $x$.  A $2$-power is called a \emph{square}.
Patterns are a generalization of powers.  A \emph{pattern}
is a non-empty word $p$ over a \emph{pattern alphabet} $\Delta$.  The
letters of $\Delta$ are called \emph{variables}.    A \emph{morphism} is
a map $h:\Sigma^* \rightarrow \Delta^*$ such that $h(xy) = h(x)h(y)$ for
all $x, y \in \Sigma$; a morphism is \emph{non-erasing} if 
$h(a) \not= \epsilon$ for all $a \in \Sigma$.  A pattern
$p$ \emph{matches} a word $w \in \Sigma^*$ if there exists a non-erasing
morphism $h : \Delta^* \to \Sigma^*$ such that $h(p) = w$.  Thus,
a word $w$ is a $k$-power if it matches the pattern $a^k$.

As mentioned above,
Anderson et al.\ \cite{ALR+09} proved that proved that the following problem
is PSPACE-complete.

\begin{quotation}
\noindent{\bf DFA/NFA PATTERN ACCEPTANCE}

\noindent INSTANCE: A DFA or NFA $M$ over the alphabet $\Sigma$ and a
pattern $p$ over some alphabet $\Delta$.

\noindent QUESTION: Does there exist $x \in L(M)$ such that
$p$ matches $x$?
\end{quotation}

In this paper we consider variations on this problem.
We consider the restricted problem where the input machine only
accepts a finite language.  We also consider the variation where the
pattern $p$ is required only to match a factor of $x$.    We show that
both of these problems are NP-complete.  We also consider the same problems
for context-free grammars; in this case the problems become PSPACE-complete.

\section{Detecting patterns in finite regular languages}

We first recall the {\bf DFA INTERSECTION} problem.  This problem is
well-known to be PSPACE-complete \cite[Problem AL6]{GJ79}.

\begin{quotation}
\noindent{\bf DFA INTERSECTION}

\noindent INSTANCE: DFAs $A_1,A_2,\ldots,A_k$, each over the alphabet $\Sigma$.

\noindent QUESTION: Does there exist $x \in \Sigma^*$ such that $x$
is accepted by each $A_i$, $1 \leq i \leq k$?
\end{quotation}

\begin{theorem}\label{fin_DFA}
The {\bf DFA INTERSECTION} problem is NP-complete if the input
DFAs only accept finite languages.
\end{theorem}

\begin{proof}
We first show that the problem is in NP.  Suppose that each $A_i$ has at most
$n$ states.  If $A_i$ accepts a finite language then it only accepts strings of
length less than $n$.  In particular any string accepted in common by all the
$A_i$'s has length less than $n$.  We can therefore guess such a string in
$O(n)$ time, and check if it is accepted by all of the $A_i$'s in $O(kn)$ time.

We show NP-completeness by reducing from 3-SAT.  Let
$\varphi$ be a boolean formula in 3-CNF, with variables $V_1, V_2,
\ldots, V_n$ and clauses $C_1, C_2, \ldots, C_m$.  We let a binary
string of length $n$ uniquely encode
a truth assignment of the variables, where $1$ denotes true and $0$ 
denotes false.
For each $1 \leq i \leq m$,
we construct a small DFA accepting exactly the strings of length $n$ that
encode an assignment of the variables $V_1, V_2, \ldots, V_n$
that satisfies clause $C_i$.  For example, if $C_1 = V_1 \vee
\overline{V_2} \vee V_4$, then our DFA would accept the strings
$$1\{0,1\}^{n-1} \cup \{0,1\}0\{0,1\}^{n-2} \cup \{0,1\}^3 1\{0,1\}^{n-4}.$$
Such a DFA can be constructed using at most $2n+1$ states.
The total number of DFA's is the
total number of clauses, and their intersection is nonempty if and only if
there is a satisfying assignment to $\varphi$.  The construction can 
clearly be carried out in polynomial time.
\end{proof}

\begin{theorem}
The {\bf NFA PATTERN ACCEPTANCE} problem is NP-complete if the input NFA
$M$ accepts a finite language.  The problem remains NP-complete if $M$ is
deterministic.
\end{theorem}

\begin{proof}
If $M$ has $n$ states, then $M$ accepts no word of length $n$ or more;
if it did, then by the pumping lemma $M$ would accept infinitely many
words.  We can therefore solve the pattern acceptance problem by guessing
a word $x$ in $L(M)$ of length less than $n$ and a morphism $h$
(also of bounded size) and verifying that $h(p) = x$ in polynomial time.

To show that the problem is NP-hard, we apply Theorem~\ref{fin_DFA}.
Given DFAs $A_1,A_2,\ldots,A_k$, each over the alphabet $\Sigma$, and
each accepting a finite language, we construct a DFA $M$ to accept
\[
L(A_1)\#\cdots L(A_k)\#,
\]
where \# is not in $\Sigma$.  The DFA $M$ accepts a $k$-power
(i.e., a word matching the pattern $a^k$) if and only if the intersection
of the $L(A_i)$'s is non-empty.
\end{proof}

\begin{theorem}
The problem ``Given an NFA $M$ and a pattern $p$, is there a non-erasing
morphism $h$ and a word $w$ in $L(M)$ such that $h(p)$ is a factor of $w$?''
is NP-complete.  The problem remains NP-complete if $M$ is deterministic.
\end{theorem}

\begin{proof}
To see that it is in NP, note that answer is always ``yes'' if $L(M)$ is
infinite (because then by the pumping lemma $L(M)$ contains $xy^*z$ for some
$x, y, z$, and if it contains arbitrarily high powers of $y$ then it contains
any pattern as a factor).  We can check if $L(M)$ is infinite in polynomial
time.  Otherwise, the size of the morphism $h$ is bounded, and we can guess
both $h$ and $w$ in polynomial time and verify that $h(p)$ is a factor of $w$.

To see the problem is NP-complete, we give a reduction from 3-SAT that
is a simple modification of the construction of Angluin
\cite[Theorem~3.6]{Ang80}.  Let $\varphi$ be a boolean formula in 3-CNF,
with variables $V_1, V_2, \ldots, V_n$ and clauses $C_1, C_2, \ldots, C_m$.
We define a pattern $p$ with variables $x_i, y_i$, $1 \leq i \leq n$ and 
$z_j, u_j$, $1 \leq j \leq m$, and $v$.  

For $1 \leq j \leq m$ and $1 \leq k \leq 3$, define
\begin{displaymath}
f(j,k) = \begin{cases}
	x_i, & \text{if the $k$'th literal in $C_j$ is $V_i$}; \\
	y_i, & \text{if the $k$'th literal in $C_j$ is $\overline{V_i}$}.
	\end{cases}
\end{displaymath}

Given $\varphi$, we define 
$$p = v^{2n+6m} \, v \, x_1 y_1 \, v \, x_2 y_2 \, v \, \cdots
\, v \, x_n y_n \, v \,  q_1 \, v \, q_2 \, v \, \cdots \, v\,  q_m 
\, v\,  z_1 u_1 \, v\,  z_2 u_2 \cdots \, v\,  z_m u_m \, v \, v^{2n+6m}$$
where $$q_j = f(j,1)f(j,2)f(j,3)z_j$$ for $ 1\leq j \leq m$, and
$$w = 0^{2n+6m} (01^3)^n (01^7)^m (01^4)^m 0 0^{2n+6m} .$$
We claim that $p$ matches a factor of $w$ if and only if $p$ matches
$w$ exactly if and only if $\varphi$ is satisfiable.
This can be established by an argument almost identical to that of
Angluin \cite[Theorem~3.6]{Ang80}.  The only difference is that
we have added $v^{2n+6m}$ to the beginning and end of $p$ and $0^{2n+6m}$
to the beginning and end of $w$ in order to enforce that $p$ matches
a factor of $w$ if and only if it matches $w$ itself.

Now we let $M$ be the $(n+2)$-state DFA that accepts the single
word $w$.  Given a 3-CNF formula $\varphi$ we can create
a DFA $M$ and pattern $p$ such that $\varphi$ is satisfiable if and only if
$p$ matches the single string accepted by $M$.
\end{proof}

\section{Detecting patterns in finite context-free languages}

We now consider the pattern acceptance problem for context-free languages.

\begin{theorem}
The following problem is undecidable: ``Given a CFG $G$, does $G$ generate
a square?''
\end{theorem}

\begin{proof}
We reduce from the Post correspondence problem.  Given an instance of
Post correspondence, say $(x_1, y_1), \ldots, (x_n, y_n)$, we 
create a CFG 
$(V, \Sigma', P, S)$ as follows:  we introduce $n+1$ new symbols
$\#, c_1, c_2, \ldots, c_n$ not in $\Sigma$, and let
$\Sigma' = \Sigma \ \cup \ \lbrace \#, c_1, c_2, \ldots, c_n \rbrace$.
Also let $V = \lbrace A, B, S \rbrace$, and let $P$ be the set of
productions
\begin{eqnarray*}
S & \rightarrow & A\# B \# \\
A & \rightarrow & x_i A c_i \ | \ x_i c_i, \ \ \ 1 \leq i \leq n \\
B & \rightarrow & y_i B c_i \ | \ y_i c_i, \ \ \ 1 \leq i \leq n  .
\end{eqnarray*}
We claim $L(G)$ contains $xx$ if and only if the PCP instance
$(x_1, y_1), \ldots, (x_n, y_n)$ has a solution.
\end{proof}

The previous problem clearly becomes decidable if $G$ only generates
a finite language.  Next we consider the computational complexity of
this restricted version of the problem.  We first consider the problem
of deciding whether two PDAs, each accepting a finite language, accept
some word $x$ in common.

Let $M$ be an TM and let $w$ be an input to $M$.  Let us suppose without
loss of generality that all halting computations of $M$ on $w$ take an
even number of steps.  It is well-known (e.g., \cite[Lemma 8.6]{HU79})
that one can construct context-free
grammars $G_1$ and $G_2$ to generate languages $L_1$ and $L_2$ consisting
of words of the form
$c_1 \# c_2^R \# \cdots \# c_{k-1} \# c_k^R \#$, where
\begin{itemize}
\item Each $c_i$ encodes a valid configuration of $M$.
\item In $L_2$, the word $c_1$ encodes the initial configuration of $M$
on input $w$ and the word $c_k$ encodes a valid accepting configuration of $M$.
\item In $L_1$ (resp. $L_2$), the configuration $c_{i+1}$ follows
from configuration $c_i$ according to the transition function of $M$ for
all odd (resp. even) $i<k$.
\end{itemize}

Suppose now that $M$ is a polynomial space bounded TM; i.e., for some
polynomial $p(n)$, $M$ uses at most $p(n)$ space on inputs of length $n$.
Consider the languages $L_1$ and $L_2$ described above, except that
now we require that each configuration $c_i$ have length at most $p(|w|)$
and that $k \leq 2^{p(|w|)}$ (since there are at most $2^{p(|w|)}$
distinct configurations in any computation of $M$ on $w$).
We can construct $G_1$ and $G_2$ as follows.  We will actually
describe the construction of a PDA $M_1$ accepting $L_1$ (the
construction for $L_2$ is similar).

First, let us observe that we can count in binary up to $2^{p(n)}$ on 
$M_1$'s stack by using $O(p(n))$ states of the finite control.  These
states simply keep track of how many bits we are pushing or popping when
incrementing the counter.

We therefore recognize a word of the form
$c_1 \# c_2^R \# \cdots \# c_{k-1} \# c_k^R \#$ as follows.
We maintain a binary counter on the stack that counts the number of
$c_i$'s that we have currently processed.  Every time we encounter a
new pair $c_i \# c_{i+1}^R$ we interrupt the current computation on the
stack---let's say we push a new temporary bottom of stack symbol onto
the stack---and we process $c_i \# c_{i+1}^R$ just as in the standard
construction.

While reading each $c_i$ (or $c_i^R$), we must also verify that the length
of $c_i$ is at most $p(n)$.  We do this by adding polynomially many states
to the finite control of $M_1$; these states are used to keep track
of the length of each $c_i$ and to verify that this length does not exceed
$p(n)$.

After verifying that $c_{i+1}$ follows from $c_i$, the stack now once more
only contains the counter recording the number of $c_i$'s
processed so far.  We can now increment this counter and continue to
process the remaining pairs in the same manner.  After reading all of
the input we verify by popping the stack that there were at most $2^{p(n)}$
$c_i$'s.

Observe that since the length of each $c_i$ in $L_1$ is at most $p(|w|)$
and since $k \leq 2^{p(|w|)}$, the language $L_1$ consists of only finitely
many words.  Furthermore, the PDAs $M_1$ and $M_2$ can be constructed
in polynomial time.

Before proceeding further we require the following lemma, which appears
to be part of the folklore.  A weaker result
was stated without proof by Meyer and Fischer
\cite[Proof of Proposition~5]{MF71}.

\begin{lemma}
\label{meyer}
Let $M$ be a PDA with $n$ states and a stack alphabet of size $s$ that
accepts by empty stack and that either pushes a single symbol onto the
stack or pops
a single symbol from the stack on each move.  If $M$ accepts a finite language,
then for any input $w$ accepted by $M$, there is an accepting computation
for which the maximum stack height is at most $sn^2$.
\end{lemma}

\begin{proof}
Consider a shortest accepting computation of $M$ on an input $w$.
Suppose that this computation has maximum stack height
$H > sn^2$ and that this height $H$ is reached after exactly $T$ steps.
For $i=1,2,\ldots,H$, let $l(i)$ denote the last time before time $T$
that the computation had stack height $i$ and let $r(i)$ denote the
first time after time $T$ that the computation had stack height $i$.
Let $C(i)=(p, A, q)$, where at time $l(i)$ the computation was in state $p$
with $A$ on top of the stack and at time $r(i)$ the computation was in
state $q$.  Since the stack height never dips below $i$ between times $l(i)$
and $r(i)$, the symbol on top of the stack at time $r(i)$ is the same
as at time $l(i)$.
There are only $sn^2$ distinct
triples $(p, A, q)$, so $C(i) = C(j)$ for some $i < j$.  We may therefore
write $w = uvwxy$ such that
\begin{itemize}
\item $u$ is the portion of the input processed after $l(i)$ steps;
\item $uv$ is the portion of the input processed after $l(j)$ steps;
\item $uvw$ is the portion of the input processed after $r(j)$ steps;
\item $uvwx$ is the portion of the input processed after $r(i)$ steps.
\end{itemize}
However, we now see that $uv^iwx^iy$ is accepted by $M$ for all positive
integers $i$.  Furthermore, $vx \neq \epsilon$---i.e., the portions of the
computation between times $l(i)$ and $l(j)$ and between times $r(j)$ and $r(i)$
do not consist entirely of $\epsilon$-transitions.  If indeed $vx = \epsilon$,
then we could obtain a shorter accepting computation of $M$ on $w$,
contradicting the assumed minimality of this computation.
Thus $M$ accepts an infinite language, a contradiction.
We conclude that the maximum stack height of a shortest accepting
computation is at most $sn^2$.
\end{proof}

We assume without loss of generality that all PDAs considered from now on
accept by empty stack and either push a single symbol or pop a single symbol
on each move.

\begin{theorem}
The following problem is PSPACE-complete:  ``Given PDAs
$A_1,A_2,\ldots,A_k$, each over the alphabet $\Sigma$, and each accepting
a finite language, does there exist $x \in \Sigma^*$ such that $x$
is accepted by each $A_i$, $1 \leq i \leq k$?''  The problem
is PSPACE-complete even when $k=2$.
\end{theorem}

\begin{proof}
To show that the problem is in PSPACE, note that each $A_i$ has
an equivalent CFG $G_i$ whose size is bounded above by a polynomial
in the size of $A_i$.  Any word generated by $G_i$ has length
bounded above by a function exponential in the size of $G_i$.

We therefore give an NPSPACE algorithm as follows.  Guess a word $w$
one symbol at a time and simulate each $A_i$ in parallel on $w$.
By Lemma~\ref{meyer}, the total space required to store the stack
contents of $A_i$ during the simulation is at most $sn^2$, where
$s$ is the size of $A_i$'s stack alphabet and $n$ is the number of states
of $A_i$.  It follows that the total space required for the parallel
simulation of the $A_i$'s is polynomial in the combined size of the $A_i$'s.
We reject on any branch of the simulation that exceeds the bound on the
stack height.  There are $s^{sn^2}$ stack configurations total, so we can
keep an $O(sn^2)$ size counter to detect if we enter an infinite loop on
$\epsilon$-transitions; if so we reject on this branch of the simulation as
well.  This non-deterministic algorithm can then be determinized by Savitch's
theorem.

To see that the problem is PSPACE-hard, it suffices to
observe that given a polynomial space bounded TM $M$ and a word $w$,
we can construct the PDAs $M_1$ and $M_2$ described above in polynomial
time.  The language $L(M_1) \cap L(M_2)$ is non-empty if and only if $M$
accepts $w$.  This completes the reduction.
\end{proof}

\begin{theorem}
\label{square}
The following problem is PSPACE-complete:  ``Given a CFG $G$ generating a
finite language, does $G$ generate a square?''
\end{theorem}

\begin{proof}
To see that the problem is PSPACE, recall that if a $G$ generates a finite
language, there is an exponential bound on the length of the words in
the language.  We now convert $G$ to a PDA $M$ in polynomial time.
Let $M$ have $n$ states and a stack alphabet of size $s$.
We wish to guess the symbols of a word $w$ of length at most exponential
in the size of $G$ and verify that $M$ accepts $ww$.  By Lemma~\ref{meyer},
the maximum stack height of $M$ on input $ww$ is at most $sn^2$.
We therefore guess a configuration $C$ of $M$ of size $O(sn^2)$
and simulate two copies of $M$ on the guessed symbols of $w$, the first
starting from the initial configuration and the second starting from the
configuration $C$.  If the first simulation ends in configuration $C$
and the second simulation ends in an accepting configuration, then $M$
accepts $ww$.  Again, we can determinize this construction by Savitch's
theorem.

To see that the problem is PSPACE-hard, it suffices to
observe that given a polynomial space bounded TM $M$ and a word $w$,
we can construct the CFGs $G_1$ and $G_2$ described above in polynomial
time.  We can then construct a CFG $G$ to generate $L(G_1)\#L(G_2)\#$.
However, $L(G)$ contains a square if and only if $M$ accepts $w$.
This completes the reduction.
\end{proof}

\begin{theorem}
The problem ``Given an CFG $G$ and a pattern $p$, is there a non-erasing
morphism $h$ and a word $w$ in $L(G)$ such that $h(p)$ is a factor of $w$?''
is PSPACE-complete.
\end{theorem}

\begin{proof}
To see that it is in PSPACE, note that answer is always ``yes'' if $L(G)$ is
infinite (because then by the pumping lemma $L(G)$ contains words with
arbitrarily high powers as factors).  We can check if $L(G)$ is infinite
in polynomial time.  If $L(G)$ is finite, then the sizes of any morphism $h$
and word $w$ such that $h(p)$ is a factor of $w$ are bounded above
by a function exponential in the size of $G$.  We can therefore guess
the symbols of $w$, the lengths of the images of $h$, and the starting
position of $h(p)$ in $w$ in polynomial space.  We may then verify that
$h(p)$ is a factor of $w$ in polynomial space by a procedure analogous
to that described in the proof of Theorem~\ref{square}, which illustrated
the method for the case of a pattern $p=xx$ (i.e., a square) matching
$w$ exactly.

We begin by converting the CFG $G$ to a PDA $M$.  We then start
guessing the symbols of $w$ and simulating $M$ on $w$.  Recall that
by Lemma~\ref{meyer}, this simulation only requires $O(sn^2)$ space.
When we reach the guessed starting location of $h(p)$ in $w$, we record
the current configuration $C_0$ of the simulation and proceed
as follows.  Let $p = p_1p_2 \cdots p_\ell$.  We begin by guessing
configurations $C_1, C_2, \ldots, C_\ell$ such that for each
$1 \leq j \leq \ell$, the simulation of $M$ goes from configuration
$C_{j-1}$ to $C_j$ upon reading $h(p_j)$.  Note that since
each configuration has size $O(sn^2)$, we can record all of these
guessed configurations in polynomial space.  We verify our guesses as follows.
For each distinct symbol $x$ occurring in $p$, let $J_x = \{j : p_j = x\}$.
For each $j \in J_x$, we simulate (in parallel for all $j$) a copy of $M$
on the symbols of a guessed word $h(x)$ (whose length we have previously
guessed) starting in configuration $C_{j-1}$.  Again, since the stack height
is bounded by a polynomial in the size of $G$, and since $|J_x|$ is at most
$|p|$, the total space required for these parallel simulations is polynomial
in the input size.  We repeat these parallel simulations for all distinct
symbols $x$ occurring in $p$.  At this stage we have guessed and verified
the occurrence of $h(p)$ in $w$; we now guess the remaining symbols of $w$
to complete the simulation.

To show PSPACE-hardness we reduce from {\bf DFA INTERSECTION}.
Suppose we are given $k$ DFAs $M_1, M_2, \ldots, M_k$, the largest
having $n$ states.  We first observe that if a word $x$ is accepted
in common by all of the $M_i$'s, there is such an $x$ of length at most $n^k$.

We now construct a CFG $G$ that generates a single squarefree word $w$ of
length at least $n^k$.  Suppose we have a uniform morphism $h$ (over a 3-letter
alphabet disjoint from those of the $M_i$'s) that generates an infinite
squarefree word.  For example, we may take $h$ to be defined by the map
\begin{eqnarray*}
0&\to&0121021201210\\
1&\to&1202102012021\\
2&\to&2010210120102
\end{eqnarray*}
(see \cite{Lee57}).  We can generate an iterate of $h$ (i.e., one of the words
$h(0), h^2(0), h^3(0), \ldots$) of length at least $n^k$ with a grammar of
size $O(k \log n)$.  So we can construct the grammar $G$ generating $w$
in polynomial time.

We then construct a CFG $G'$ that generates all prefixes of words in $L(G)$.
Next we convert $G'$ to a PDA $N$ (in polynomial time).
Given the PDA $N$ and a DFA $M_i$, we can perform the standard
construction to obtain a PDA $A_i$ accepting the perfect shuffle
of $L(N)$ and $L(M_i)$.  We do this so that the strings accepted by
$A_i$ are all squarefree.

Now we convert all of the $A_i$'s into CFGs $B_i$ (we can do this in
polynomial time).  Next we construct a grammar $C$ that accepts the
language $L = L(B_1) \# L(B_2) \# \cdots \# L(B_k) \#$.

The size of the grammar $C$ accepting $L$ is just the sum of the sizes of
the $B_i$'s, so it remains polynomial in the combined sizes of the
original DFAs.  Now, since $L(B_i)$ consists only of squarefree strings,
a word $u$ in $L$ contains a $k$-power as a factor if and only if $u$ is
itself a $k$-power.  (Note also that $L$ is a finite language.)

Recall that if there is an $x$ accepted by all the $M_i$'s, there is
such an $x$ of length at most $n^k$.  Since the words in $L(G')$ have
length at most $n^k$, there is such an $x$ if and only if a string $z$
formed by the perfect shuffle of a word $x$ and a word in $L(G')$
is accepted by all of the $A_i$'s.  But this is true if and only if the
string $(z\#)^k$ is generated by $C$.  This in turn is true if and
only if $C$ generates a word with a $k$-power as a factor.

The entire construction can be done in polynomial time.  This
completes the reduction.
\end{proof}

\section{Acknowledgment}

    We thank Giovanni Pighizzini for suggesting how to improve the
bound from $(sn)^2$ to $sn^2$ in Lemma~\ref{meyer}.

\end{document}